\documentclass{llncs}
\usepackage{botex}
\usepackage{verbatim}
\usepackage[normalem]{ulem}

\newcommand{\lc}{\mathrm{len}}
\newcommand{\vol}{\mathrm{vol}}
\newcommand{\work}{\ensuremath{\mathrm{work}}}
\newcommand{\awork}{\ensuremath{\mathrm{\hat{w}}}}

\usepackage{xcolor}
\newcommand{\vcom}[1]{}
\newcommand{\awcom}[1]{}

\begin{document}

\title{Feasibility Tests for Recurrent Real-Time Tasks in the Sporadic DAG Model}
\author{%
Vincenzo Bonifaci\inst{1} \and
Alberto Marchetti-Spaccamela\inst{2} \and \\
Sebastian Stiller\inst{3} \and
Andreas Wiese\inst{4}}
 \institute{IASI-CNR, Rome, Italy \\
 \and
 Sapienza Universit\`{a} di Roma, Italy \\
 \and
 TU Berlin, Germany \\
 \and 
 Max Planck Institute for Informatics, Saarbr\"ucken, Germany
}%

\maketitle

\newcommand{\gsys}{\ensuremath{\mathcal{G}}}
\newcommand{\tsys}{\ensuremath{\mathcal{T}}}
\newcommand{\sinf}{\ensuremath{S_{\infty}}}

\begin{abstract}
A model has been proposed in \cite{Baruah:2012:a} for representing recurrent precedence-constrained tasks to be executed on multiprocessor platforms, where each recurrent task is modeled by a directed acyclic graph (DAG), a period, and a relative deadline. Each vertex of the DAG represents a sequential job, while the edges of the DAG represent precedence constraints between these jobs.  All the jobs of the DAG are released simultaneously and have to be completed within some specified relative deadline.  The task may release jobs in this manner an unbounded number of times, with successive releases occurring at least the specified period apart. The feasibility problem is to determine whether such a recurrent task can be scheduled to always meet all deadlines on a specified number of dedicated processors.

The case of a single task has been considered in \cite{Baruah:2012:a}. The main contribution of this paper is to consider the case of multiple tasks. We show that EDF has a speedup bound of $2-1/m$, where $m$ is the number of processors. Moreover, we present polynomial and pseudopolynomial schedulability tests, of differing effectiveness, for determining whether a set of sporadic DAG tasks can be scheduled by EDF to meet all deadlines on a specified number of processors.
\end{abstract}

\section{Introduction}
\label{sec:intro}

\vcom{I changed the intro here and there (not much)}

The sporadic task model is a well-known model to represent real-time systems based on  a finite number of independent recurrent processes or tasks, each of which may generate an unbounded sequence of jobs. 
Determining how multiple recurrent tasks can be scheduled on a shared uni- or multi-processor platform is one of the traditional objects of study in real-time scheduling theory. 
Different formal models have been proposed for representing such recurrent tasks; 
these models differ from one another in the restrictions they place on the jobs that may be generated by a single task (see, for example, \cite{Dertouzos:1974,Liu:1969:b,Liu:1969:a,Liu:1973,Stigge:2011:a}).

The technological evolution of  processor manufacturing is moving  away from increasing clock frequencies to increasing the number of cores per processor.  This is a continuing trend, with no immediate end in sight.
The presence of large core-counts offers new opportunities for executing more computation-intensive workloads in real time. 
It is still unclear how the resulting massively parallel multicore CPUs will be structured. For example, it is not clear whether all the cores will be identical, or there will be different specialized  cores to realize different functions, and/or   whether some cores will  be dedicated to certain functionalities, with the rest being general-purpose processors.

These and similar questions are still unanswered;  however, it seems likely that in the near future an execution environment will allow for the possibility of having more expressive task models than the relatively simple recurrent task models considered thus far in the real-time scheduling literature. We refer to 
~\cite{Saifullah:2011,Stigge:2011:a,Stigge:2011:b} and to references therein for a thorough discussion of the models. 
We observe that an important characteristic of the more expressive models is to allow for partial parallelism within a task, as well as for precedence constraints between different parts of the task.  
%

In this paper, we continue the study of a parallel task model, the {\em sporadic DAG model\/}, that was introduced in \cite{Baruah:2012:a} and that  considers   the preemptive scheduling of a recurrent task. The task is modeled as a directed acyclic graph (DAG) $G=(V,E)$ which executes upon a platform consisting of $m$ identical processors that are dedicated to the exclusive use of this particular task. The task repeatedly emits a \emph{dag-job}, which is a set of precedence-constrained sequential jobs.

More precisely, in \cite{Baruah:2012:a} each vertex $v\in V$  of the DAG corresponds to a sequential job, and is characterized by a worst-case execution time (WCET) $p_v$.   Each (directed) edge of the DAG represents a precedence constraint:  if $(v,w) \in E$ is a (directed) edge in the DAG, then the job corresponding to vertex $v$ must complete execution before the job corresponding to vertex $w$ may begin execution.  Groups of jobs that are not constrained (directly or indirectly) by precedence constraints in such a manner may execute in parallel if there are processors available for them to do so.

When a dag-job is released by the task, it is assumed that all $|V|$ of the corresponding jobs become available for execution simultaneously, subject to the precedence constraints.  During any given run the task may release an unbounded sequence of dag-jobs; all $|V|$ jobs that are released at some time-instant $t$ must complete execution by time-instant $t+D$, where $D$ is the (relative) deadline parameter of the task.  A minimum interval of duration $T$ must elapse between successive releases of dag-jobs, where $T$ is the period of the task. 

\paragraph{Previous results.}
It is  known~\cite{Ullman:1975} that the preemptive scheduling of a given collection of precedence-constrained jobs (i.e.,  a DAG) on a multiprocessor platform is NP-hard in the strong sense; this intractability result is easily seen to hold for the sporadic DAG model as well.  

Much of the research described in  \cite{Baruah:2012:a} is  concerned with dealing with the case $D>T$, which in the case of a single DAG is the more interesting case.   First it is shown that the ``synchronous arrival sequence", in which successive dag-jobs are released exactly the period $T$ time-units apart, does not necessarily correspond to the worst-case behavior of a  sporadic DAG task; hence, we cannot determine schedulability properties by simply studying this one behavior of the task.  
Furthermore, \cite{Baruah:2012:a} also considers the Earliest Deadline First (EDF) scheduling~\cite{Liu:1973,Dertouzos:1974} of a sporadic DAG task on identical multiprocessors. It is shown that EDF has a \emph{speedup bound} (this metric is formally defined in Section~\ref{sec:defns}) no larger than 2 for scheduling a sporadic DAG task.  The paper also presents two different schedulability tests for determining whether EDF can schedule a given sporadic DAG task upon a specified identical multiprocessor to meet all deadlines.  These tests have different run-time complexity --- one has polynomial run-time while the other has run-time pseudopolynomial in the representation of the task --- and effectiveness (as quantified, again, by the {speedup bound\/} metric).

\paragraph{This paper.}
The main  limitation of \cite{Baruah:2012:a} is that a single DAG task is considered. The major contribution of this paper is to consider the case of  multiple tasks, where each task is specified by a different DAG. 

The remainder of this paper is organized as follows.  In Section~\ref{sec:defns}, we formally define the notation and terminology used in describing our task model.  We also formalize the concepts of feasibility, schedulability, and schedulability testing, and the speedup bound metric.
 In Section~\ref{sec:edf} we present a speedup bound for EDF (which extends a result of~\cite{Phillips:2002}). We present and analyze a pseudopolynomial time EDF schedulability test in Section~\ref{sec:test}, and a  simple sufficient polynomial-time condition in Section~\ref{sec:simple-test}.  

\section{Model and definitions}\label{sec:defns}
%
In the \emph{sporadic DAG} model, a task $\tau_i$ ($i=1,\ldots, n$) is specified as a 3-tuple $(G_i, D_i, T_i)$, where $G_i$ is a vertex-weighted directed acyclic graph (DAG), and $D_i$ and $T_i$ are positive integers.
\begin{itemize}
\item The DAG $G_i$ is specified as $G_i=(V_i,E_i)$, where $V_i$ is a set of vertices and $E_i$ a set of directed edges between these vertices (it is required that these edges do not form any oriented cycle). Each $v\in V_i$ denotes a sequential operation (a ``job").  Each job $v\in V_i$ is characterized by a processing time $p_v \in \Nat$, also known as {\em worst-case execution time\/} or WCET.   The edges represent dependencies between the jobs: if $(v_1,v_2)\in E_i$ then job $v_1$ must complete execution before job $v_2$ can begin execution.   (We say a job becomes {\em eligible\/} to execute once all its predecessor jobs have completed execution.)
\item A {\em period\/} $T_i \in \Nat$.  A {\em release\/} or arrival of a {\em dag-job\/} of the task at time-instant $t$ means that all $|V_i|$  jobs $v\in V_i$ are released at time-instant $t$.  The period denotes the minimum amount of time that must elapse between the release of successive dag-jobs: if a dag-job is released at $t$, then the next dag-job cannot be released prior to time-instant $t+T_i$.
\item A {\em deadline\/} $D_i \in \Nat$. If a dag-job is released at time-instant $t$ then all $|V_i|$ jobs that were released at $t$ must complete execution by time-instant $t+D_i$. 
\end{itemize}

Throughout this paper we assume that the input consists of a \emph{task system} $\tsys=(\tau_1,\tau_2,\ldots,\tau_n)$, a collection of $n$ sporadic DAG tasks.  
If $D_i>T_i$, the task $\tau_i$ may release a dag-job prior to the completion of its previously-released dag-jobs. We do 
\emph{not} require that all jobs of a dag-job complete execution  before jobs of the next dag-job can start executing. 
\medskip

Some additional notation and terminology:
\begin{itemize}
\item A {\em chain\/} in the sporadic DAG task $\tau_i$ is a sequence of vertices $v_1, v_2,\ldots,v_k$ such that $(v_j,v_{j+1})$ is an edge in $G_i$, $1 \le j <k$.  The \emph{length} of this chain is defined to be the sum of the WCETs of all its vertices: $\sum_{j=1}^k p_{v_j}$.   

\item   We denote by $\lc(G_i)$ the length of the longest chain in $G_i$.  Note that $\lc(G_i)$ can be computed in time linear in the number of vertices and the number of edges in $G_i$, by first obtaining a topological order of the vertices of the graph and then running a straightforward dynamic program.
\item  We define $\vol(G_i)=\sum_{v \in V_i} p_v$.  That is, $\vol(G_i)$ is the total WCET  of each dag-job.  It is evident that $\vol(G_i)$ can be computed in time linear in the number of vertices in $G_i$.
\end{itemize}
\medskip

\paragraph{Feasibility and schedulability.}
Since the period parameter $T_i$ of the sporadic DAG task $\tau_i$ specifies the minimum, rather than exact, duration that must elapse between the release of successive dag-jobs, a task system may generate infinitely many different collections of dag-jobs. A task system $\tsys$ is said to be \emph{feasible} on $m$ speed-$s$ processors if a valid schedule exists on $m$ speed-$s$ processors for every collection of dag-jobs that may be generated by the task system. 
A task system is said to be \emph{EDF-schedulable} on $m$ speed-$s$ processors if EDF meets all deadlines when scheduling any collection of  dag-jobs that may be generated by the task system on $m$ speed-$s$ processors.

\paragraph{Speedup bounds.}  The problem of testing feasibility of a given DAG task system is highly intractable (\ccnp-hard in the strong sense~\cite{Ullman:1975}) even when $n=1$. It is therefore highly unlikely that we will be able to design efficient algorithms for solving the problem exactly, and our objective is therefore to come up with efficient algorithms that solve the problem {\em approximately\/}. 
In this paper, we present \emph{EDF-schedulability tests with speedup $s$} for some values $s$. These are tests which either guarantee that a system is EDF-schedulable on $m$ machines with speed $s$, or prove that the system is infeasible on $m$ machines of unit speed. 
The value $s$ is called the \emph{speedup bound} of the test and it is the metric we will use for quantifying the quality of the approximation. 

Sometimes one is unable to provide a test with a speedup bound. In that case a positive answer of the test may only be \emph{sufficient} to guarantee the EDF-schedulability of a DAG task system on $m$ unit speed processors, without any guarantee in the case where the test yields a negative answer. 

\section{Analysis of EDF for a collection of jobs}\label{sec:edf}
This section considers an arbitrary collection $J$ of dag-jobs, including, but not restricted to, any collection that may be generated by a given task system $\tsys$.  




We are given $m$ identical parallel processors, and the jobs in $J$ are revealed online over time. Each job $j$ is characterized by a release date $r_{j} \in \Nat_0$, an absolute deadline $d_{j} \in \Nat$, a processing time $p_{j} \in \Nat$, and a set of previous jobs $J_{j}$ which are exactly the jobs which have to be finished before $j$ becomes eligible (the \emph{predecessors} of $j$).
We assume that if $p_j$ is a predecessor of $p_k$ then $r_j=r_k$ and $d_j=d_k$. We call such a collection of jobs $J$ a \emph{normal} collection of jobs. Observe that every collection of jobs generated by a sporadic DAG task system is normal, since all jobs that constitute a certain dag-job have identical release date and deadline. 

At any time, the EDF scheduler processes the $m$ jobs with minimum deadline which are currently available (breaking ties arbitrarily).
A job $j$ is \emph{available} \emph{at time }$t$ if $t\ge r_{j}$ and all jobs in $J_{j}$ have been completed, while $j$ is not yet completed. We denote the length of a time interval $I$ by $|I|$. 

\begin{lemma}
\label{lem:EDF-speeds}
Consider a normal collection $J$ of jobs and let $\alpha \ge 1$. Then at least one of the following holds:  
\begin{enumerate}
\item all jobs in $J$ are completed within their deadline under EDF on $m$ processors of
speed $\alpha$, or 
\item $J$ is infeasible under arbitrarily many processors of unit
speed, or
\item there is an interval $I$ such that any feasible schedule for $J$ must finish
more than $(\alpha m-m+1) \cdot |I|$ units of work within $I$.
\end{enumerate}
\end{lemma}

\begin{proof}
Suppose that both (1) and (2) do not hold, that is, under EDF on $m$ speed-$\alpha$ processors some job $j$ fails its deadline $d_{j}$, and $J$ is feasible if we are given a large enough number of processors. When given infinitely many (or, say, $|J|$) processors of unit speed, a simple greedy schedule is optimal: just allocate one processor to each job and schedule each job as early as possible. Denote by $\sinf$ such a greedy schedule; observe that $\sinf$ starts and ends processing jobs always at integral time points. Note, that $\sinf$ is globally optimal in the sense that at any point in time and for any job it has processed at least as much of that job as any algorithm on unit speed processors.

Without loss of generality, we can assume that there is no job $j'$ in the instance with $d_{j'}>d_{j}$ (otherwise, since $J$ is normal the removal of $j'$ does not affect EDF nor $\sinf$). Let $t^{*}$ denote the latest point in time before which EDF has processed at least as much of \emph{every }job as $\sinf$. Such a time exists, since $t=0$ satisfies this property. Also, it must hold that $t^{*}<d_{j}$. We claim that within $I:=[t^{*},d_{j}]$ EDF finishes more than $(\alpha m-m+1) \cdot |I|$ units of work, hence $\sinf$ finishes at least the same amount of work during $I$ (by construction of $I$) and hence \emph{every} feasible schedule has to finish more than $(\alpha m-m+1) \cdot |I|$ units of work during $I$. 

Denote by $X$ the total length of the intervals within $I$ where in the EDF schedule all $m$ processors are busy. Define $Y:=|I|-X$. We distinguish two cases. First assume that $\alpha\cdot Y\ge|I|$. Denote by $Y_{1},...,Y_{k}\subseteq I$ all subintervals of $I$ where not all processors are busy. We define $t'$ such that $\alpha\cdot|[t^{*},t']\cap\bigcup_{i}Y_{i}|=\left\lceil t^{*}\right\rceil -t^{*}$.
During all timesteps within $[t^{*},t']\cap\bigcup_{i}Y_{i}$ all jobs are available for EDF which are scheduled by $\sinf$ during $[t^{*},\left\lceil t^{*}\right\rceil ]$.
Since during all these timesteps EDF does not use all processors and runs the processors with speed $\alpha$, by time $t'$ it has processed at least as much of every job as $\sinf$ by time $\left\lceil t^{*}\right\rceil $.
Now define timesteps $t_{i}$, $i=0,...,d_{j}-\left\lceil t^{*}\right\rceil $ such that $\alpha\cdot|[t^{*},t_{i}]\cap\bigcup_{i}Y_{i}|=\left\lceil t^{*}\right\rceil -t^{*}+i$ for each $i$. We prove by induction that up to time $t_{i}$ EDF
has processed as much of every job as $\sinf$ by time $\left\lceil t^{*}\right\rceil +i$.
The case $i=0$ was proven above. Now suppose that the claim is true for some value $i$. Then at each timestep during $[t_{i},t_{i+1}]\cap\bigcup_{i}Y_{i}$ all jobs are available for EDF that $\sinf$ works on during $[\left\lceil t^{*}\right\rceil +i,\left\lceil t^{*}\right\rceil +i+1]$.
Since during all these timesteps EDF does not use all processors and runs the processors with speed $\alpha$, by time $t_{i+1}$ it has processed at least as much of every job as $\sinf$ by time $\left\lceil t^{*}\right\rceil +i+1$.
By induction the claim is true for $i^{*}=d_{j}-\left\lceil t^{*}\right\rceil $ and hence at time $\left\lceil t^{*}\right\rceil +i^{*}=d_{j}$ EDF has finished as much of every job as $\sinf$. This yields a contradiction since we assumed that $\sinf$ is feasible and EDF is not.

Now assume that $\alpha\cdot Y<|I|$. Hence, in the interval $I$ EDF finishes at least 
\begin{eqnarray*}
\alpha m\cdot X+\alpha\cdot Y & = & \alpha m\cdot(|I|-Y)+\alpha\cdot Y\\
 & = & \alpha m\cdot|I|-\alpha mY+\alpha\cdot Y\\
 & > & \alpha m\cdot|I|-m\cdot|I|+|I|\\
 & = & (\alpha m-m+1) \cdot |I|
\end{eqnarray*}
units of work, and by construction of $I$, any feasible solution has to finish during the interval $I$ all work that EDF finishes during $I$.
\qed
\end{proof}

The above lemma implies the following theorem if we choose $\alpha=2 - 1/m$.

\begin{theorem}
Any normal collection of jobs that is feasible on $m$ processors of unit speed is  EDF-schedulable on $m$ processors of speed $2-1/m$.
\end{theorem}

\begin{proof}
Since we assumed the instance to be feasible, it is in particular feasible on a sufficiently high number of processors of unit speed. Also, the instance admits a valid schedule which finishes in any interval $I$ at most $m\cdot|I|$ units of work.
Note that if $\alpha=2-1/m$ then $(\alpha m-m+1) \cdot |I| = (2m-1-m+1) \cdot |I| = m |I|$.
Hence, Lemma~\ref{lem:EDF-speeds} implies that EDF finishes all jobs by their respective deadline.
\qed
\end{proof}

Since every collection of jobs generated by a sporadic DAG task system is normal, we obtain the following corollary. 
\begin{corollary}
Any task system that is feasible on $m$ processors of unit speed is EDF-schedulable on $m$ processors of speed $2-1/m$. 
\end{corollary}
Notice that the bound is tight: examples are known (even without precedence constraints) of feasible collections of jobs that are not EDF-schedulable unless the  speedup is at least $2-1/m$ \cite{Phillips:2002}.

\section{A pseudopolynomial test}\label{sec:test}

In the following we present a pseudopolynomial test based  on a characterization of the work that a feasible instance requires.

Recall the definition of $\sinf$ from the proof of Lemma \ref{lem:EDF-speeds}. For a sequence of jobs $J$ and an interval $I$, we denote by $\work^J(I)$ 
the amount of work done by $\sinf$ during $I$ on the jobs in $J$ whose deadline is in $I$. 

\begin{definition}
Given a sporadic DAG task system $\tsys$, let $\text{gen}(\tsys)$ be the set of job sequences that may be generated by $\tsys$, and define 

$$\work_\tsys(t) := \sup_{J \in \text{gen}(\tsys)} \sup_{t_0 \ge 0} \work^J([t_0,t_0+t]).$$  
$$ \lambda_\tsys := \sup_{t \in \Nat} \frac{\work_\tsys(t)}{t}. $$
\end{definition}

The following lemma shows that a bound on $\work_\tsys(t)$ allows one to show that EDF is feasible with some speedup.

\begin{lemma}
\label{lem:work}
Let $\tsys$ be a sporadic DAG task system. Let $\eps \ge 0$ and suppose that $\work_\tsys(t)\le(1+\epsilon)mt$
for any $t \in \Nat$ and that $\tsys$ is feasible on a large enough number of unit-speed processors\footnote{Observe that \tsys\ is feasible on a large enough number of unit-speed processors if and only if $\lc(G_i) \le D_i$ for all $i=1,\ldots,n$.}. Then \tsys\ is EDF-schedulable on $m$ processors of speed $2-1/m+\epsilon$.
\end{lemma}
\begin{proof}
Suppose that EDF fails on some job sequence $J \in \mathrm{gen}(\tsys)$ when running at speed $2-1/m+\epsilon$.
Then by Lemma~\ref{lem:EDF-speeds} there is an interval $I$ in which any feasible schedule must finish more than $(\alpha m-m+1) \cdot |I| = (2m-1+\epsilon m-m+1) |I| = (1+\epsilon)m|I|$ units of work. This contradicts that $\work_\tsys(|I|)\le(1+\epsilon) m |I|$. 
\qed
\end{proof}

Therefore, in order to approximately test the feasibility of $\tsys$ it suffices to estimate $\lambda_\tsys$. 

\begin{lemma}
Let $\eps \ge 0 $ and $\hat{\lambda}_\tsys$ be such that $\lambda_\tsys/(1+\eps) \le \hat{\lambda}_\tsys \le  \lambda_\tsys$. Assume that $\tsys$ is feasible on a large enough number of unit-speed processors. Then
\begin{enumerate}
\item if $\hat{\lambda}_\tsys > m $, \tsys\ is infeasible on $m$ unit speed processors; 
\item if $\hat{\lambda}_\tsys \le m $, \tsys\ is EDF-schedulable on $m$ speed-$(2-1/m+\eps)$ processors. 
\end{enumerate}
\end{lemma}
\begin{proof}
In case (1), $\lambda_\tsys \ge \hat{\lambda}_\tsys > m$, therefore there is a job collection $J \in \mathrm{gen}(\tsys)$ and an interval $I$ such that $\work^J(I) > m|I|$, hence $\tsys$ is not feasible on $m$ unit speed machines. 

In case (2), $\lambda_\tsys \le (1+\eps) \hat{\lambda}_\tsys \le (1+\eps)m$, therefore Lemma \ref{lem:work} applies.   
\qed
\end{proof}

\vcom{I fixed the statement of the corollary}

\begin{corollary}
\label{cor:lambda}
Let $\eps \ge 0$. 
A $(1+\eps)$-approximation algorithm for $\lambda_\tsys$ yields an EDF-schedulability test for $\tsys$ with speedup $2-1/m+\eps$. 
\end{corollary}

\paragraph{Approximation of $\lambda_\tsys$.}

We now show how to efficiently estimate $\lambda_\tsys$. Since the tasks $\tau_1,\ldots,\tau_n$ of \tsys\ are independent of each other, we can equivalently write
$$ \lambda_\tsys = \sup_{t \in \Nat} \frac{\sum_{i=1}^n \work_i(t)}{t} $$
where $ \work_i(t)$ is the maximum amount of work that may be done by $\sinf$ on jobs of task $\tau_i$ in an interval of length $t$. 
This maximum is achieved when the deadline of some job of $\tau_i$ coincides with the rightmost endpoint of the interval, and the other jobs of $\tau_i$ are released as closely as possible. That is, if the interval is (without loss of generality) $[t_0,t_0+t]$, then there is
\begin{itemize}
\item one job with release date $t_0+t-D_i$ and deadline $t_0+t$, 
\item one job with release date $t_0+t-D_i-T_i$ and deadline $t_0+t-T_i$, 
\item one job with release date $t_0+t-D_i-2 T_i$ and deadline $t_0+t-2T_i$, 
\item \dots
\item in general, one job with release date $t_0+t-D_i-k T_i$, up to a $k$ such that $t_0+t-(k+1) T_i \le t_0$ (more jobs would not contribute to the amount of work done by $\sinf$ during $[t_0,t_0+t]$). 
\end{itemize}

As a consequence, $\work_i(t)$ is piecewise linear as a function of $t$, with a number of pieces that is proportional to $|V_i| \cdot t/T_i$, as each dag-job is responsible for at most $|V_i|$ pieces.

\begin{lemma}
\label{lem:work:bounds}
For any task $\tau_i=(G_i,D_i,T_i)$, 
\begin{align} 
\work_i(t) & \ge \max \left( \floor{\frac{t+T_i-D_i}{T_i}}, 0 \right) \cdot \vol(G_i), \label{eq:work:lb} \\ 
\work_i(t) & \le \ceil{\frac{t}{T_i}} \cdot \vol(G_i). \label{eq:work:ub}
\end{align}
\end{lemma}
\begin{proof}
\eqref{eq:work:lb}: there can be as many as $\floor{(t+T_i-D_i)/T_i}$ releases of $\tau_i$-dag-jobs in an interval of length $t$ whose release date and deadline fall within the interval; each of them contributes $\vol(G_i)$ to the work function.  

\eqref{eq:work:ub}: there cannot be more than $\ceil{t/T_i}$ releases of $\tau_i$-dag-jobs in an interval of length $t$ whose deadline falls within the interval. These dag-jobs are the only ones that contribute a positive amount of work. 
\qed
\end{proof}

Since the number of pieces of $\work_i(t)$ grows with $t$, it is not clear how to handle this function efficiently. Therefore, we approximate $\work_i(t)$ by a function $\awork_i(t)$ defined as follows: 
\begin{align*}
\awork_i(t) := \begin{cases} 
\work_i(t) & \text{ if } t \le T_i/\eps + (1+1/\eps) D_i \\
\frac{t-D_i}{T_i}\, \vol(G_i) & \text{ if } t > T_i/\eps + (1+1/\eps) D_i.  
\end{cases}
\end{align*}

\begin{lemma}
The piecewise linear function $\awork_i$ has $O(\frac{1}{\eps} \cdot |V_i| \cdot (1+\frac{D_i}{T_i}))$ many pieces. 
\end{lemma}
\begin{proof}
Immediate from the definition of $\awork_i$ and the properties of $\work_i$. 
\qed
\end{proof}

\begin{corollary}
\label{cor:pieces}
Let $\awork(t) := \sum_{i=1}^n \awork_i(t)$. 
The piecewise linear function $\awork$ has $O(\frac{1}{\eps} \cdot \sum_{i=1}^n |V_i| \cdot \max_{i=1}^n (1+\frac{D_i}{T_i}))$ many pieces. 
\end{corollary}

\begin{lemma}
For all $i=1,\ldots,n$ and all $t \in \Nat$, 
\begin{equation*} 
\frac{1}{1+\eps} \work_i(t) \le \awork_i(t) \le \work_i(t). 
\end{equation*}
\end{lemma}
\begin{proof}
First observe that $\work_i(t) \ge \awork_i(t)$, since for all $t > T_i/\eps + (1+1/\eps)D_i$, by \eqref{eq:work:lb}, 
$$ 
\frac{\work_i(t)}{\vol(G_i)} \ge 
\floor{\frac{t+T_i-D_i}{T_i}} \ge \frac{t+T_i-D_i}{T_i} -1 = \frac{t-D_i}{T_i} = \frac{\awork_i(t)}{\vol(G_i)}. $$
Moreover, using \eqref{eq:work:ub}, 
\begin{align*} 
\frac{\work_i(t)}{\awork_i(t)} \le \frac{\ceil{t/T_i}}{\frac{t-D_i}{T_i}} \le \frac{t/T_i + 1}{t/T_i - D_i/T_i} = \frac{t+T_i}{t-D_i} \le \frac{(D_i+T_i)/\eps + D_i+T_i}{(D_i+T_i)/\eps+D_i-D_i} = 1+\eps.
\end{align*}
\qed
\end{proof}

\begin{corollary}
\label{cor:apx}
For all $t \in \Nat$, 
$\frac{1}{1+\eps} \work(t) \le \awork(t) \le \work(t). $
\end{corollary}

\begin{lemma}
\label{lem:piecewise}
Let $f: \Nat \to \Nat$ be a piecewise linear function with $K$ pieces and assume we can compute $\lim_{t \to \infty} f(t)/t$. Then the value $\sup_{t \in \Nat} f(t)/t$ can be found by evaluating $f$ in $O(K)$ points. 
\end{lemma}
\begin{proof}
Let $[a,b]$ be a piece of $f$, that is, a maximal interval in which $f$ is linear. Then $f(t)/t$ is monotone in $[a,b]$, so that $\max(f(a)/a,f(b)/b) \ge f(t)/t$ for all $t \in [a,b]$. Therefore, to compute $\sup_{t \in \Nat} f(t)/t$ it suffices to compute the value of $f$ in $K+1$ points (one of these ``points'' is $t=\infty$). 
\qed
\end{proof}

\vcom{I fixed the statement of the theorem}

\begin{theorem}
Let $\eps > 0$. There is a pseudopolynomial time EDF-schedulability test with speedup $2-1/m+\eps$. 
\end{theorem}
\begin{proof}
After combining Corollary \ref{cor:lambda}, Corollary \ref{cor:pieces}, Corollary \ref{cor:apx} and Lemma \ref{lem:piecewise}, it only remains to show that each $\awork_i(t)$ can be evaluated in pseudopolynomial time for any $t$. This is clear from the definition of $\awork_i$ when $t > T_i/\eps + (1+1/\eps)D_i$. When $t \le T_i/\eps + (1+1/\eps)D_i$, notice that there can be $O(1+D_i/T_i)$ dag-jobs that contribute only partially (less than $\vol(G_i)$) to $\awork_i(t)$. For each of them, the exact amount of contributed work can be computed in polynomial time. 
\qed
\end{proof}

\section{A simple sufficient condition for EDF-schedulability}\label{sec:simple-test}
We complement the result of the previous section with a sufficient condition for EDF-schedulability that can be easily checked in polynomial time.

Given  a sporadic DAG task system, w.l.o.g.~we assume that DAGs $G_i$ are ordered according to nondecreasing $D_i$
(breaking ties arbitrarily).

\begin{theorem}\label{thm:polytimeImproved}
Assume a sporadic DAG task system satisfies the following properties: 
\begin{enumerate}
\item $\lc(G_k)\le D_k/3$, $k=1,2, \ldots, n$,
\item for each $k$, $k=1,2, \ldots, n$, 
$$ \sum_{i: T_i \le D_k} \vol(G_i) /T_i  + \sum_{i: T_i > D_k} \vol(G_i) / D_k \le (m +1/2)/3.$$  
\end{enumerate}
Then the system is EDF-schedulable on $m$ unit-speed processors. 
\end{theorem}
\begin{proof}
Suppose by contradiction that EDF fails to meet some deadline while scheduling some sequence of dag-jobs released by a sporadic task $\tau_k$. 
Let $j$ be the first job of DAG $G_k$ that misses its deadline $d_{j}$. W.l.o.g.
we assume that there are no jobs with a deadline later than $d_{j}$.
Consider the interval $I:=[r_{j},d_{j})$. 
Denote by $X$ the total amount of time during $I$ where all processors are busy. Let $Y:=(d_{j}-r_{j})-X=D_k-X$, i.e., $Y$ denotes the total amount of time in $I$ during which not all processors are busy. 

We  first observe that $Y\le D_k/3$. This follows from the observation that whenever a processor is idle, EDF must be executing a job belonging  the longest chain of the last activation of $G_k$ and hence $Y \le \lc(G_k)$, which is assumed to be at most $D_k/3$.  

Condition $Y\le  D_k/3$ implies that  $X\ge 2D_k/3$.  Now since the total amount of execution occurring over the interval $I$ is greater or equal to $(mX+Y)$, we conclude that 
%
the total work done by EDF during $I$ is greater or equal to $(2m+1) D_k /3 $.

Now recall \eqref{eq:work:ub} and observe that the total amount of work due in $I$ is bounded above by

\begin{eqnarray*}
& & \sum_{i: T_i \le D_{k}} \left\lceil \frac{D_k}{T_i}\right\rceil \vol(G_i)  + 
 \sum_{i: T_i > D_{k}}\vol(G_i) \\
& &  \le  2D_k   \left(\sum_{i: T_i \le D_{k}} \vol(G_i) /T_i  + \sum_{i: T_i > D_{k}}\vol(G_i)/ D_k \right) \\
& &  \le \frac{2m+1}{3} D_k 
\end{eqnarray*}
where we have used the fact that $\ceil{x} \le 2x$ when $x \ge 1$. 
This contradicts the assumption that EDF fails and completes the proof of the theorem.
\qed
\end{proof}


\bibliographystyle{abbrv-plus}
\bibliography{refs}

\end{document}